\documentclass[a4paper, 10pt, conference]{ieeeconf}      

\IEEEoverridecommandlockouts                              
\let\labelindent\relax
\usepackage{amsmath,amssymb,amsfonts}  
\usepackage{enumitem,graphicx,xspace,color}
\usepackage{xcolor}

\usepackage{dsfont}
\usepackage{theorem}
\theoremstyle{plain}
\newtheorem{theorem}{Theorem}
\newtheorem{definition}{Definition}
\newtheorem{proposition}{Proposition}
\newtheorem{lemma}{Lemma}
\newtheorem{corollary}{Corollary}

{\theorembodyfont{\upshape}
\newtheorem{example}{Example}
\newtheorem{remark}{Remark}}

\newcommand{\R}{\mathbb{R}}
\newcommand{\N}{\mathbb{N}}
\newcommand{\1}{\mathds{1}}
\newcommand{\0}{\mathbb{0}}

\newcommand{\diag}[1]{\text{diag}\left\{#1\right\}}
\newcommand{\sign}[1]{\text{sign}\left(#1\right)}
\newcommand{\range}[1]{\mathcal{R}(#1)}
\newcommand{\card}[1]{ \text{card}(#1) }
\renewcommand{\ker}[1]{\mathcal{N}(#1)}
\newcommand{\abs}[1]{ \left\lvert #1 \right\rvert }
\newcommand{\norm}[1]{ \| #1 \|}
\newcommand{\spectrum}[1]{\mathrm{sp}(#1)}
\newcommand{\real}[1]{\mathrm{Re}[#1]}
\newcommand{\imag}[1]{\mathrm{Im}[#1]}
\newcommand{\G}{\mathcal{G}}
\newcommand{\edgeSet}{\mathcal{E}}
\newcommand{\nodeSet}{\mathcal{V}}

\newcommand{\Ldagger}{L^\dagger}
\newcommand{\half}{\frac{1}{2}}

\newcommand{\corank}[1]{\text{corank}(#1)}
\newcommand{\rank}[1]{\text{rank}(#1)}
\newcommand{\vspan}[1]{\text{span}(#1)}

\newcommand{\epos}{\overset{\vee}{>}}
\newcommand{\PF}{\mathcal{PF}\xspace}

\newcommand{\Rtot}{R_\mathrm{tot}}

\title{\LARGE \bf On the properties of Laplacian pseudoinverses}

\author{Angela Fontan and Claudio Altafini
\thanks{Work supported in part by a grant from the Swedish Research Council (grant n. 2020-03701) and from the Swedish ELLIIT program.}
\thanks{A. Fontan and C. Altafini are with the Division of Automatic Control,
	Department of Electrical Engineering, Link\"{o}ping University, SE-58183 Link\"{o}ping,	Sweden, E-mail:  $\{$angela.fontan, claudio.altafini$\}$@liu.se.}}

\makeatletter
\def\endthebibliography{%
	\def\@noitemerr{\@latex@warning{Empty `thebibliography' environment}}%
	\endlist
}
\makeatother

\begin{document}
\bstctlcite{IEEEexample:BSTcontrol}
\maketitle
\thispagestyle{empty}
\pagestyle{empty}

\begin{abstract}
The pseudoinverse of a graph Laplacian is used in many applications and fields, such as for instance in the computation of the effective resistance in electrical networks, in the calculation of the hitting/commuting times for a Markov chain and in continuous-time distributed averaging problems. In this paper we show that the Laplacian pseudoinverse is in general not a Laplacian matrix but rather a signed Laplacian with the property of being an eventually exponentially positive matrix, i.e., of obeying a strong Perron-Frobenius property. 
We show further that the set of signed Laplacians with this structure (i.e., eventual exponential positivity) is closed with respect to matrix pseudoinversion. 
This is true even for signed digraphs, and provided that we restrict to Laplacians that are weight balanced also stability is guaranteed.

\end{abstract}

\section{INTRODUCTION}

For a network or a networked system, the Laplacian matrix is a fundamental object that captures information about e.g. connectivity and spectrum \cite{Chung1997,Agaev2005Spectra}, as well as properties of the dynamics that live on the graph \cite{OlfatiSaberMurray2004,AltafiniLini2015,BronskiDeville2014,PanShaoMesbahi2016}.
Associated to the Laplacian is also a Laplacian pseudoinverse, typically a Moore-Penrose pseudoinverse, which has also been used extensively to describe graph-related quantities. 
For instance it is used to build an effective resistance matrix for the graph, a distance measure that exploits the analogy between graphs and electrical networks \cite{KleinRandic1993,XiaoGutman2003,GhoshBoydSaberi2008}, and to compute hitting/commuting times in Markov chains \cite{Chandra1996,Palacios2001,Boley2011,VanMieghem2017}.
It is also used to estimate the $ \mathcal{H}_2 $ norm in networked dynamical systems \cite{YoungScardoviLeonard2010,YoungScardoviLeonard2011,lindmark2020investigating}.

If we consider a graph with nonnegative edge weights, it is well-known that the Laplacian $L$ is an M-matrix (i.e., a matrix with nonpositive off-diagonal entries, such that $ -L$ is marginally stable, 
see below for proper definitions). 
It is also easy to show that the Laplacian pseudoinverse does not belong to the same class, not even when the graph is undirected. 
Consider for instance the following Laplacian matrix 
\[
L = \begin{bmatrix}
    0.8  & -0.7 &  -0.1 \\
   -0.7 &   0.9  & -0.2 \\
   -0.1 &  -0.2 &   0.3
   \end{bmatrix}.
 \]
 Its pseudoinverse is 
\[
 L^\dagger = \begin{bmatrix}    0.773  &  
 0.048 &  -0.821 \\
 0.048  &  0.628 &  -0.676 \\
 -0.821 &  -0.676  &  1.498
 \end{bmatrix} 
\]
which has an anomalous sign in the (1,2) entry, even though it has the same stability properties of $ L$. 

The aim of this paper is to investigate the algebraic properties of Laplacian pseudoinverses. 
Even though $ L^\dagger $ is not an M-matrix, it has nevertheless most of the properties of M-matrices, most notably it obeys to a strong Perron-Frobenius property: the pair formed by the eigenvalue $0$ and eigenvector $ \1 = \begin{bmatrix} 1 & \ldots & 1 \end{bmatrix}^T $ is the ``dominant'' eigenpair for $ - L^\dagger$ (just like it is for $ -L $) in spite of the presence of positive off-diagonal entries in $ L^\dagger$.
Such matrices are called {\em eventually exponentially positive} in the linear algebra literature \cite{Noutsos2006,NoutsosTsatsomeros2008,JohnsonTarazaga2004,AltafiniLini2015}.
We show in the paper that this argument can be extended to signed Laplacians, i.e., Laplacians associated to signed digraphs\footnote{As signed Laplacian here we use the so-called ``repelling Laplacian'' in the terminology of \cite{Shi2019Dynamics}, see Section~\ref{sec:preliminaries} for a precise definition.}: the pseudoinverse of an eventually exponentially positive signed Laplacian is an eventually exponentially positive signed Laplacian. 
In other words, the class of eventually exponentially positive signed Laplacians is closed with respect to pseudoinversion. 

Under the assumption of edge weight balance, such class is also closed with respect to stability, i.e., $ -L $ and $ -L^\dagger $ eventually exponentially positive are also marginally stable (and of corank $1$).

When we restrict further the class from weight balanced $L$ to normal $L$, then it coincides also with the class of Laplacians and Laplacian pseudoinverses whose symmetric part is positive semidefinite of corank 1. 
Such restriction is particularly useful in the context of effective resistance, which, being a distance, has to be symmetric. 
For normal signed Laplacians we obtain a natural way to extend the notion of effective resistance to digraphs, alternative to the definitions already appeared in the literature, see e.g. \cite{YoungScardoviLeonard2016a}.

\section{PRELIMINARIES}
\label{sec:preliminaries}

\subsection{Linear algebraic preliminaries}
Given a matrix $A= [a_{ij}] \in \R^{n\times n}$, $A\ge 0$ means element-wise nonnegative, i.e., $a_{ij}\ge 0$ for all $i,j=1,\dots,n$, while $A>0$ means element-wise positive, i.e., $a_{ij}>0$ for all $i,j=1,\dots,n$.
The spectrum of $A$ is denoted $\spectrum{A}= \{\lambda_1(A),\dots,\lambda_n(A)\}$, where $\lambda_i(A)$, $i=1,\dots,n$, are the eigenvalues of $A$.
In this paper we use the ordering $\real{\lambda_1(A)}\le \real{\lambda_2(A)} \le \dots\le \real{\lambda_n(A)}$, where $\real{\lambda_i(A)}$ indicates the real part of $\lambda_i(A)$.
The spectral radius of $A$ is the smallest real nonnegative number such that $\rho(A)\ge \abs{\lambda_i(A)}$ for all $i=1,\dots,n$ and $\lambda_i(A)\in \spectrum{A}$.
A matrix $A$ is called \textit{Hurwitz stable} if $\real{\lambda_n(A)}<0$, and \textit{marginally stable} if $\real{\lambda_n(A)}=0$ is a simple root of the minimal polynomial of $A$.

A matrix $A$ is called positive semidefinite (psd) if $x^T A x = x^T \frac{A+A^T}{2} x\ge 0$ $\forall x\in \R^n$ and it is called positive definite (pd) if $x^T A x = x^T \frac{A+A^T}{2} x > 0$ $\forall x\in \R^n\setminus \{0\}$.

A matrix $A$ is called irreducible if there does not exist a permutation matrix $P$ s.t. $P^T A P$ is block triangular.

A matrix $B$ is called a Z-matrix if it can be written as $B=s I -A$, where $A\ge 0$ and $s>0$, and it is called a M-matrix if, in addition, $s\ge \rho(A)$, which implies that all the eigenvalues of $B$ have nonnegative real part.
If $s>\rho(A)$ then $B$ is nonsingular and $-B$ is Hurwitz stable.
If $s = \rho(A)$ then $B$ is singular, and if $A$ is irreducible then $-B$ is marginally stable.

If $A$ is a singular matrix, the Moore-Penrose pseudoinverse of A, denoted $A^\dagger$, is the unique $n\times n$ matrix that satisfies $A A^\dagger A =A$, $A^\dagger A A^\dagger =A^\dagger$, $(A^\dagger A)^T = A^\dagger A$, and $(A A^\dagger)^T = A A^\dagger$.
A singular matrix $ A $ is said to have index 1 if the range of $ A $, $ \range{A} $, and the
kernel of $ A $, $\ker{A}$, are complementary subspaces, i.e., $\range{A}\cap \ker{A} = \emptyset$. For index 1 singular matrices, the Drazin inverse and the group inverse coincide. 
A singular M-matrix has always index 1 \cite{Meyer2000}.
A matrix is {\em normal} if it commutes with its transpose: $ A A^T = A^T A $.
A matrix $ A $ is said an \textit{EP matrix} (Equal Projector, also called a \textit{range
symmetric matrix} \cite{Meyer2000}) if $\ker{A}=\ker{A^T}$ (and hence $\range{A}=\range{A^T}$). 
EP matrices generalize normal matrices, and like normal matrices have many equivalent characterizations, see \cite{Meyer2000}. 
For instance an EP matrix $ A $ is such that $ A $ commutes with its Moore-Penrose pseudoinverse $ A^\dagger$. 
If $ A $ is an EP-matrix, then $\exists\,U$ orthogonal such that
\begin{equation*}
A= U \begin{bmatrix} 0&0\\0&B\end{bmatrix} U^T
\end{equation*}
with $ B $ nonsingular of dimension $  r = \rank{A} $. Singular EP matrices have index 1, and for them the Moore-Penrose pseudoinverse, the Drazin inverse and the group inverse coincide.

\subsection{Signed graphs}
Let $\G(A) = (\nodeSet, \edgeSet, A)$ be the (weighted) digraph with vertex set $\nodeSet$ ($\card{\nodeSet}=n$), $\edgeSet =\nodeSet \times \nodeSet$ and adjacency matrix $A= [a_{ij}] \in \R^{n\times n}$: $a_{ij}\in \R\setminus\{0\}$ iff $(j,i)\in \edgeSet$. 
Since each edge of the digraph is labeled by a sign (i.e., $\sign{a_{ij}}=\pm 1$), $\G(A)$ is called a signed digraph. In the particular case where $A\ge 0$, the digraph $\G(A)$ is called nonnegative.
For digraphs $\G(A)$ which are strongly connected and without self-loops, the matrix $A$ is irreducible with null-diagonal.

The weighted in-degree and out-degree of node $i$ are denoted $\sigma_i^{\mathrm{in}} = \sum_{j=1}^n a_{ij}$ and $\sigma_i^{\mathrm{out}} = \sum_{j=1}^n a_{ji}$, respectively.
The (signed) \textit{Laplacian} of a graph $\G(A)$ is the (in general non-symmetric) matrix $L= \Sigma -A$ where $\Sigma = \diag{\sigma_1^{\mathrm{in}},\dots,\sigma_n^{\mathrm{in}}}$.
This definition of signed Laplacian corresponds to the so-called ``repelling Laplacian'' in the terminology of \cite{Shi2019Dynamics}.
By construction, this Laplacian is a singular matrix with $\ker{L}= \vspan{\1}$. 
However, $-L$ need not be marginally stable and its symmetric part $L_s=\frac{L+L^T}{2}$ need not be positive semidefinite, as we show in the examples in Section~\ref{sec:results}.
Moreover, $L$ irreducible (or, $\G(A)$ strongly connected) need not imply $\corank{L}=1$. For instance, consider a complete, undirected, signed graph $\G(A)$ whose Laplacian is 
\begin{equation*}
L= \begin{bmatrix} 
3&-1&-1&-1\\
-1&1&1&-1\\
-1&1&1&-1\\
-1&-1&-1&3
\end{bmatrix}.
\end{equation*}
It is $\spectrum{L}=\{0,0,4,4\}$ and $\1, [0,1,-1,0]^T\in \ker{L}$, i.e., $L$ is marginally stable of corank $2$.
Similarly, $\corank{L}=1$ need not imply $L$ irreducible.

A digraph $\G(A)$ is \textit{weight balanced} if in-degree and out-degree coincide for each node, i.e., $\sigma_i^{\mathrm{in}} = \sum_{j=1}^n a_{ij} = \sum_{j=1}^n a_{ji} = \sigma_i^{\mathrm{out}}$ for all $i=1,\dots,n$. 
As we show in Lemma~\ref{lemma:irreducible}, $\corank{L}=1$ and $L$ weight balanced imply $L$ irreducible.

\subsection{Eventual exponential positivity}
\begin{definition}
	A matrix $A \in \R^{n\times n}$ has the (strong) Perron-Frobenius property if $\rho(A)$ is a simple positive eigenvalue of $A$ s.t. $\rho(A) > \abs{\lambda(A)}$ for every $ \lambda(A) \in \spectrum{A}$, $ \lambda(A)\ne \rho(A)$, and $\chi$, the right eigenvector relative to $\rho(A)$, is positive.
\end{definition}
The set of matrices which possess the Perron-Frobenius property will be denoted $\PF$, and it is known (see e.g. \cite[Thm 8.4.4]{HornJohnson2013}) that irreducible nonnegative matrices are part of this set.
However, it has been shown (see \cite{Noutsos2006}) that matrices having negative elements can also possess this property, provided that they are eventually positive.
\begin{definition}
	A matrix $A \in \R^{n\times n}$ is called \textit{eventually positive} (denoted $A\epos 0$) if $\exists k_0 \in \N$ s.t. $A^k>0$ for all $k\ge k_0$.
\end{definition}
\begin{theorem}\cite[Thm 2.2]{Noutsos2006}\label{thm:evenpos}
	Let $A\in \R^{n\times n}$. Then the following statements are equivalent:
	\begin{enumerate}
		\item Both $A, A^T \in \PF$;
		\item $A\epos 0$;
		\item $A^T\epos 0$.	 
	\end{enumerate}
\end{theorem}

\begin{definition}
	A matrix $A \in \R^{n\times n}$ is called \textit{eventually exponentially positive} if $\; \exists \, t_0 \in \N$ s.t. $e^{At}>0$ for all $t\ge t_0$.
\end{definition}
\begin{lemma}\cite[Thm 3.3]{NoutsosTsatsomeros2008}\label{lemma:evenexppos}
	A matrix $A\in \R^{n\times n}$ is eventually exponentially positive if and only if $\exists \,d\ge 0$ s.t. $A+d I \epos 0$.
\end{lemma}

\subsection{Kron reduction for undirected networks}
\label{sec:Kron_reduction}
Consider an undirected, connected, weighted graph $\G(A)=(\nodeSet,\edgeSet,A)$ with adjacency matrix $A= [a_{ij}] \in \R^{n\times n}$.
Let $\alpha \subset \{1,\dots,n\}$ (with $\card{\alpha}\ge 2$) and $\beta = \{1,\dots,n\} \setminus \alpha$ be a partition of the node set $\nodeSet=\{1,\dots,n\}$. 
After an adequate permutation of its rows and columns, the Laplacian $L$ of the graph $\G(A)$ can be rewritten as $L = \begin{bmatrix} L[\alpha] & L[\alpha,\beta] \\L[\beta,\alpha] & L[\beta]\end{bmatrix}$, where we denote $L[\alpha,\beta]$ the submatrix of $L$ determined by the index sets $\alpha$ and $\beta$, and $L[\alpha]:= L[\alpha, \alpha]$ the principal submatrix of $L$ determined by the index set $\alpha$.

If $L[\beta]$ is nonsingular, the Schur complement of $L[\beta]$ in $L$ is given by $L/L[\beta] := L[\alpha]- L[\alpha,\beta] L[\beta]^{-1} L[\beta,\alpha]$.
In the context of electrical networks, where $\alpha$ and $\beta$ are referred to as boundary (or terminal) and interior nodes, this procedure is denoted Kron reduction (see e.g. \cite{DorflerBullo2013,DorflerPorcoBullo2018}) and it yields a matrix $L_r:=L/L[\beta] $, denoted Kron-reduced matrix, which is still a Laplacian of a weighted, undirected graph $\G_r$ (see \cite{DorflerBullo2013} for details and properties of $L_r$).

If $\G(A)$ is signed, when $\alpha$ is chosen as the set of nodes incident to edges with negative weight it is shown in \cite{Chen2016a} that $L[\beta]$ is positive definite and that $L$ is psd of corank $1$ if and only if $L_r$ is psd of corank $1$. 

\section{PSEUDOINVERSE OF EVENTUALLY EXPONENTIALLY POSITIVE LAPLACIANS}
\label{sec:results}
In this section we study the connection between the marginal stability and eventual positivity of the Laplacian $L$ and of its pseudoinverse $\Ldagger$.
If the network is undirected and signed, or if the network is directed, signed and weight balanced, we show that $-L$ is eventually exponentially positive if and only if $\Ldagger$ is eventually exponentially positive.

\subsection{Directed signed network case}
Assume that the graph $\G(A) = (\nodeSet,\edgeSet,A) $ is directed and without loops, which means that the adjacency matrix $A$ is with null diagonal.

When the graph is weight balanced, the Laplacian is a EP-matrix since $\ker{L}=\ker{L^T}=\vspan{\1}$. In this case, it is shown in \cite{Altafini2019Investigating} that $-L$ is eventually exponentially positive if and only if $-L$ is marginally stable (of corank $1$). 
In addition, if the Laplacian is a normal matrix, then eventual exponential positivity of $L$ is equivalent to that of its symmetric part.
\begin{theorem}
\label{thm:Altafini2019}
	Consider a signed digraph $\G(A)$ such that the corresponding Laplacian $L$ is weight balanced. Then, the following conditions are equivalent:
	\begin{enumerate}[label=(\roman*)]
	    \item $-L$ is eventually exponentially positive;
        \item $-L$ is marginally stable of corank $1$.
	\end{enumerate}
	Furthermore, if $L$ is normal then (i) and (ii) are equivalent to
	\begin{enumerate}[resume,label=(\roman*)]
	    \item $L_s = \frac{L+L^T}{2}$ is psd of corank $1$.
	\end{enumerate}
\end{theorem}
\begin{proof}
(i)$\Longleftrightarrow$(ii) See \cite[Corollary 2]{Altafini2019Investigating}.
   
(ii)$\Longleftrightarrow$(iii): If $L$ is normal, then there exists an orthonormal matrix $ U$ such that $ L = U D U^T $, where, if  $ \mu_1, \ldots, \mu_k $ are the real eigenvalues of $ L $ and $ \nu_1\pm i \omega_1, \ldots , \nu_\ell \pm i \omega_\ell $ are its complex conjugate eigenvalues:
\begin{equation*}
	D = \mu_1 \oplus \cdots \oplus \mu_k \oplus
			\begin{bmatrix}  \nu_1 & \omega_1 \\-\omega_1 & \nu_1 \end{bmatrix} \oplus \cdots \oplus
			\begin{bmatrix}  \nu_\ell & \omega_\ell \\-\omega_\ell & \nu_\ell
			\end{bmatrix}
\end{equation*}
where $\oplus$ indicates direct sum.
If follows that $ L_s = \frac{L + L^T}{2} =  \frac{1}{2} U (D + D^T) U^T $ and therefore that $ {\rm Re}[\lambda_i (L) ] = \lambda_i (L_s) $. 
\end{proof}

Observe that Theorem~\ref{thm:Altafini2019} does not explicitly assume that $\G(A)$ is strongly connected. However, as we will show later in Lemma~\ref{lemma:irreducible}, any of the conditions (i) or (ii) implies that $L$ is irreducible (i.e., that $\G(A)$ is strongly connected).

\begin{remark}
Corollary 2 of \cite{Altafini2019Investigating} claims that the equivalence (ii) $ \Longleftrightarrow $ (iii) is valid in the more general case of weight balance $L$. Unfortunately that result is not true as the following Example~\ref{example:Ls_nonpsd} shows.
A complication arises for instance from the fact that for $ L$ weight balanced but not normal $ L_s $ may acquire negative diagonal elements even if $ -L $ is marginally stable. $ L_s $ with negative diagonal elements obviously cannot be psd. However, even when $ L_s $ has positive diagonal it is not guaranteed to be psd, see Example~\ref{example:Ls_nonpsd_posdiag}.
\end{remark}

\begin{example}\label{example:Ls_nonpsd}
In correspondence of 
\begin{equation*}
L=\begin{bmatrix} 
0.15&0&0&-0.15\\
-0.23&0.15&0.15&-0.07\\
0.01&-0.12&-0.03&0.14\\
0.07&-0.03&-0.12&0.08
\end{bmatrix}
\end{equation*}
it is $\spectrum{L}=\{0, 0.0901\pm 0.199 i ,0.169\}$, i.e., $-L$ is marginally stable of corank $1$. Moreover, $L\1 = L^T\1=0$ and, for $d>0.2647$, $B= d I -L \epos 0$. 
However, $\spectrum{L_s}=\{-0.0402, 0, 0.1248 ,0.2655\}$, i.e., $L_s$ is not psd.
\end{example}

\begin{example}\label{example:Ls_nonpsd_posdiag}
For
\begin{equation*}
L=\begin{bmatrix} 
0.23&0&-0.28&0.05\\
-0.01&0.03&0.02&-0.04\\
0.05&-0.03&0.04&-0.06\\
-0.27&0&0.22&0.05
\end{bmatrix}
\end{equation*}
it is $\spectrum{L}=\{0, 0.1443\pm 0.1859 i ,0.0514\}$, i.e., $-L$ is marginally stable of corank $1$. Moreover, $L\1 = L^T\1=0$ and, for $d>0.1919$, $B= d I -L \epos 0$.
However, $\spectrum{L_s}=\{-0.0446, 0, 0.0404 ,0.3441\}$, i.e., $L_s$ is not psd.
\end{example}

As already mentioned, for signed Laplacians, irreducibility does not imply corank 1. When we have weight balance, however, the opposite is true.
\begin{lemma}\label{lemma:irreducible}
Let $\G(A)$ be a signed digraph with Laplacian $L$. 
If $-L$ is eventually exponentially positive or if $L$ is weight balanced and of corank $1$, then $L$ is irreducible.
\end{lemma}
\begin{proof}
In both statements assume, by contradiction, that $L$ is reducible, i.e., there exists a permutation matrix $P$ s.t. $P^T L P = \begin{bmatrix} L_{11} & L_{12} \\ 0 & L_{22} \end{bmatrix}$. 

Assume that $-L$ is eventually exponentially positive, i.e., $\exists\, d\ge 0$ s.t. $B=d I-L\epos 0$ (see Lemma~\ref{lemma:evenexppos}). 
Then $B$ is also reducible, since $P^T B P = \begin{bmatrix} d I-L_{11} & -L_{12} \\ 0 & d I-L_{22} \end{bmatrix}$. It follows that $(P^T B P)^k = \begin{bmatrix} (d I-L_{11})^k & \ast \\ 0 & (d I-L_{22})^k \end{bmatrix}$ for all $k\ge 1$, i.e., $P^T B P$ is not eventually positive and, consequently, $B$ is not eventually positive.

Assume that $L$ is weight balanced of corank $1$. Then $L\1=L^T\1=0$ implies that $0 \in \spectrum{L_{11}^T}=\spectrum{L_{11}}$ and that $0\in \spectrum{L_{22}}$. 
Consequently, $L$ is not of corank $1$.
\end{proof}

\begin{remark}
    For a signed digraph $\G(A)$ it holds that if $L_s$ is psd of corank $1$ then $L$ is EP (see \cite{LewisNewman1968}) and hence weight balanced, and $-L$ is marginally stable of corank $1$. Therefore, in Theorem~\ref{thm:Altafini2019}, the assumption that the Laplacian $L$ is a normal matrix is sufficient to prove that $L_s$ is psd of corank $1$ but not necessary. For example, for
    \begin{equation*}
    L=\begin{bmatrix} 
        1&1&-1&-1\\ -1&1&0&0\\
        -1&-1&2&0\\  1&-1&-1&1
        \end{bmatrix},
    \end{equation*}
    which is not normal, it is $\spectrum{L}=\{0, 1.5\pm 1.323 i ,2\}$, i.e., $-L$ is marginally stable of corank $1$, and $\spectrum{L_s}=\{0, 0.7192,1.5, 2.7808\}$, i.e., $L_s$ is psd of corank $1$.
\end{remark}
We now show that the same statements of Theorem~\ref{thm:Altafini2019} hold also for the pseudoinverse $\Ldagger$ of $L$.
Moreover, we show that $-L$ is eventually exponentially positive (and marginally stable) if and only if $-L^\dagger$ is.
These results are summarized in the following theorem.
\begin{theorem}\label{thm:directed}
	Let $\G(A)$ be a directed signed network such that the corresponding Laplacian $L$ is weight balanced. Let $\Ldagger$ be the weight balanced pseudoinverse of $L$. 
	Then, the following conditions are equivalent:
	\begin{enumerate}[label=(\roman*)]
	    \item $-L$ is eventually exponentially positive;
	    \item $-\Ldagger$ is marginally stable of corank $1$;
	    \item $-\Ldagger$ is eventually exponentially positive.
	\end{enumerate}
	Furthermore, if $L$ (equivalently, $\Ldagger$) is normal then (i)$\div$(iii) are equivalent to
	\begin{enumerate}[resume,label=(\roman*)]
	    \item $\Ldagger_s = \frac{\Ldagger+(\Ldagger)^T}{2}$ is psd of corank $1$.
	\end{enumerate}
\end{theorem}
The proof of Theorem~\ref{thm:directed} relies on some considerations and propositions that we state first.

Since a weight balanced $L$ (of corank $1$) is a EP-matrix, its left and right orthogonal projectors onto $\range{L}$ are identical and given by $\Pi = I-J$, with $J=\frac{\1 \1^T}{n}$. Furthermore it is $\lim_{t\to \infty} e^{-Lt} = J$.
The following properties for $ J $ can be found in \cite{Bullo2020Lectures} (\cite{BenIsrealGreville2003,Meyer2000}) or computed straightforwardly.
\begin{lemma}\label{lemma:properties_J}
	The matrix $J$ has the following properties:
	\begin{enumerate}
		\item $J^k = J$ $\forall k\in \N$ which implies that $(I-J)^k=(I-J)$ $\forall k\in \N$;
		\item $J L = LJ = 0$ which implies that $e^{-(L+J)}=e^{-L}e^{-J}$ and $Je^{-L}=e^{-L} J=J$;
		\item $e^{-Jt}=I-J+Je^{-t}$ which implies that $Je^{-Jt}=e^{-Jt}J = Je^{-t}$.
	\end{enumerate}
\end{lemma}
We have the following properties for the Laplacian pseudoinverse.
\begin{lemma}\label{lemma:properties_Ldagger}
	If $ L $ is weight balanced and of corank $1$, then $ \Ldagger $ is weight balanced and of corank $1$. For it
	\begin{gather}
		L \Ldagger =\Ldagger L = \Pi 	\label{eqn:pseudoinv_wb_1}\\
		\Ldagger \1 = (\Ldagger)^T \1 = 0 \label{eqn:pseudoinv_wb_2}\\
		\Ldagger \Pi = \Pi \Ldagger = \Ldagger \label{eqn:pseudoinv_wb_3}\\
		\Ldagger = (L+\gamma J)^{-1} -\frac{1}{\gamma} J \quad \forall \gamma \ne 0.\label{eqn:pseudoinv_wb_4}
	\end{gather}
	Furthermore, if $L$ is normal then $\Ldagger$ is normal.
\end{lemma}
\begin{proof}
    Assume that $L$ is weight balanced.
	Eq.~\eqref{eqn:pseudoinv_wb_1}-\eqref{eqn:pseudoinv_wb_4} are all well-known for $ L $ symmetric, and follow easily also for EP matrices. They are proven here only for sake of completeness.
	Eq.~\eqref{eqn:pseudoinv_wb_2} is a consequence of $ L $ commuting with $\Ldagger$.
	As for eq.~\eqref{eqn:pseudoinv_wb_2}, from $(\Ldagger L)^T=\Ldagger L$	it follows that $L^T (\Ldagger)^T \1 = 0$. Since $ L $ is irreducible, $L^T v \ne 0$ for $v\ne c\1$ ($c \in \R$), hence it must be $(\Ldagger)^T\1 = 0$ or $\1^T \Ldagger = 0$, i.e., $\Ldagger$ has $\1$ as left eigenvector relative to $0$. The proof for the right eigenvector is identical.
	Concerning eq.~\eqref{eqn:pseudoinv_wb_3}, from $\Ldagger\1 = 0$ it is
	$\Ldagger \Pi = \Ldagger(I-\frac{\1 \1^T}{n})=\Ldagger$, and similarly for $\Pi \Ldagger=\Ldagger$.
	For eq.~\eqref{eqn:pseudoinv_wb_4}, since $L+\gamma J$ is nonsingular, as in \cite{DorflerBullo2013}, it is enough to show the following:
	\begin{align*}
		(L+\gamma J)(\Ldagger+\frac{1}{\gamma} J)
		&=L \Ldagger + \gamma J \Ldagger + \frac{1}{\gamma} LJ + J^2
		\\&=\Pi + J = I-J+J=I,
	\end{align*}
	where we have used the properties of Lemma~\ref{lemma:properties_J}.
	
	Then, $\ker{L}=\ker{L^T}=\ker{\Ldagger}=\ker{(\Ldagger)^T}=\vspan{\1}$ and \eqref{eqn:pseudoinv_wb_4} imply that $\Ldagger$ is weight balanced of corank $1$. Notice that irreducibility of $L$ and $\Ldagger$ follows from Lemma~\ref{lemma:irreducible}.
	
	Finally, we need to show that if $L$ is normal then $\Ldagger$ is normal. $L$ normal, $J$ symmetric and $L J= L^T J= JL =J L^T=0$ imply $L+\gamma J$ normal, which means that $(L+\gamma J)^{-1}$ is also normal. Since $J$ is symmetric (hence normal) and satisfies the properties of Lemma~\ref{lemma:properties_J}, to show that $\Ldagger$ is normal it is sufficient to observe that $(L+\gamma J)^{-1} J = \frac{1}{\gamma} J = J (L+\gamma J)^{-1}$.
\end{proof}

We are now ready to prove Theorem~\ref{thm:directed}.\\
\begin{proof}
(i)$\Longrightarrow$(ii) To show marginal stability of $-\Ldagger$, denote $\lambda_i(L)$ the eigenvalues of $ L $, of eigenvectors $\1, v_2,\dots v_n$. 
Using Theorem~\ref{thm:Altafini2019}, since $-L$ is eventually exponentially positive then $-L$ is also marginally stable of corank $1$, meaning that $0 = \lambda_1(L) < \real{\lambda_2(L)} \le \dots \le  \real{\lambda_n(L)} $. 
Consider eq.~\eqref{eqn:pseudoinv_wb_4} of Lemma~\ref{lemma:properties_Ldagger}.
Choosing $\gamma\ne 0$, since $ J $ is the orthogonal projection onto $\ker{L}= \ker{L^T} = \mathrm{span}(\1)$,	the effect of adding $\gamma J$ to $ L $ is only to shift the $ 0 $ eigenvalue to $\gamma$, while $\lambda_2(L),\dots, \lambda_n(L)$ are unchanged (see \cite[Thm 2.4.10.1]{HornJohnson2013}). 
For the nonsingular $L+\gamma J$ the inverse $(L+\gamma J)^{-1}$ has eigenvalues $\frac{1}{\gamma},\frac{1}{\lambda_2(L)},\dots, \frac{1}{\lambda_n(L)}$ of eigenvectors $\1, v_2,\dots v_n$. From orthogonality, $(L+\gamma J)^{-1}-\frac{1}{\gamma}J$ only shifts the $ \frac{1}{\gamma} $ eigenvalue back to the origin without touching the other eigenvalues.
	
(i)$\Longrightarrow$(iii) Assume that $-L$ is eventually exponentially positive, that is, $-L$ is marginally stable of corank $1$ (see Theorem~\ref{thm:Altafini2019}).
Then $-\Ldagger$ is also marginally stable of corank $1$, see Lemma~\ref{lemma:properties_Ldagger} and proof (i)$\Longrightarrow$(ii). 
To prove that $-\Ldagger$ is eventually exponentially positive, we can use Theorem~\ref{thm:Altafini2019}. The proof is here reported for completeness.
In particular, from Lemma~\ref{lemma:properties_Ldagger}, we know that $-\Ldagger$ is marginally stable with $0 = \lambda_1(\Ldagger) < \real{\lambda_2(\Ldagger)} \le \dots \le  \real{\lambda_n(\Ldagger)} $ and with $\1$ as left/right eigenvector for $0$.
If we choose $d > \max_{i=2,\dots,n} \frac{\abs{\lambda_i(\Ldagger)}^2}{2 \real{\lambda_i(\Ldagger)}}$, then $B = dI - \Ldagger$ has $\rho(B)=d$ as a simple eigenvalue of eigenvector $\1$ and so does $B^T$.
Hence $B, B^T \in \PF$, or, from Theorem~\ref{thm:evenpos}, $B \epos 0$, i.e., B is eventually positive. Hence from Lemma~\ref{lemma:evenexppos} $\Ldagger$ is eventually exponentially positive. 

(iii)$\Longrightarrow$(i) Since $\Ldagger$ is weight balanced of corank $1$ with $\mathrm{span}(\1) = \ker{\Ldagger} = \ker{(\Ldagger)^T}$, it is itself a signed Laplacian.
The argument can be proven in a similar way as the opposite direction, observing that $L=(\Ldagger)^\dagger$.
\smallskip

(iv) Assume now that $L$ is normal or, equivalently, that $\Ldagger$ is normal (see Lemma~\ref{lemma:properties_Ldagger}).
Since $L$ normal implies $L$ weight balanced, the statements (i), (ii), and (iii) are still equivalent.
To show the equivalence with (iv) it is sufficient to apply Theorem~\ref{thm:Altafini2019} on $\Ldagger$ since $\Ldagger$ is itself a normal signed Laplacian of corank $1$. 
\end{proof}

The following corollary, characterizing the class of eventually exponentially positive Laplacian matrices, follows directly from Theorem~\ref{thm:directed}.
\begin{corollary}\label{cor:class_exppos_directed}
The class of eventually exponentially positive, weight balanced Laplacian matrices is closed under the pseudoinverse operation.

The class of eventually exponentially positive, normal Laplacian matrices is closed under the pseudoinverse and the symmetrization operation (the latter intended as the operation of taking the symmetric part). 
\end{corollary}

\begin{remark}
    Notice that the operations of pseudoinverse and of symmetrization do not commute, i.e., $\Ldagger_s\!=\!\frac{\Ldagger+(\Ldagger)^T}{2} \ne (L_s)^\dagger \!= \!(\frac{L+L^T}{2})^\dagger$, even in the case of a normal Laplacian $L$.
    Indeed, let $L=U D U^T$ with $U$ orthonormal and $D$ as in the proof of Theorem~\ref{thm:Altafini2019}. Without lack of generality, assume that the first column of $U$ is $\1\frac{1}{\sqrt{n}}$, which means that $D = 0 \oplus \bar D$ where $\bar D = \mu_2 \oplus \dots \oplus \mu_k \oplus \begin{bmatrix} \nu_1 & \omega_1 \\ -\omega_1 & \nu_1 \end{bmatrix}\oplus \dots \oplus \begin{bmatrix} \nu_\ell & \omega_\ell \\ -\omega_\ell & \nu_\ell \end{bmatrix}$ is nonsingular (here $\oplus$ denotes the direct sum). Then
    \begin{gather*}
        L_s =  U \begin{bmatrix} 0 & 0\\0 & \frac{\bar D+\bar D^T}{2}\end{bmatrix} U^T ,
        \quad 
        \Ldagger = U \begin{bmatrix} 0 & 0\\0 &  \bar D^{-1}\end{bmatrix} U^T,
        \\
        \Ldagger_s =  U \begin{bmatrix} 0 & 0\\0 & \frac{\bar D^{-1}+\bar D^{-T}}{2}\end{bmatrix} U^T,
      \end{gather*}
    and 
    \begin{align*}
        (L_s)^\dagger 
        &= U \begin{bmatrix} 0 & 0\\0 & (\frac{\bar D+\bar D^T}{2})^{-1}\end{bmatrix} U^T \ne \Ldagger_s.
    \end{align*}
\end{remark}

\begin{example}\label{example:Ls_nonpsd_cont}
For the Laplacian $L$ of Example~\ref{example:Ls_nonpsd}, we obtain 
\begin{equation*}
\Ldagger = \begin{bmatrix} 
2.25&-1.86&-0.19&-0.19\\
-1.42&1.58&-5.64&5.47\\
1.92&0.47&4.36&-6.75\\
-2.75&-0.19&1.47&1.47
\end{bmatrix}.
\end{equation*}
It is $\spectrum{\Ldagger}=\{0, 1.8888\pm 4.1709 i , 5.8891\}$ (notice that $\lambda_i(\Ldagger)=\frac{1}{\lambda_i(L)}$, $i=2,3,4$), i.e., $-\Ldagger$ is marginally stable of corank $1$. Moreover, $\Ldagger\1 = (\Ldagger)^T\1=0$ and, for $d>5.5495$, $B= d I -\Ldagger \epos 0$. 
However, $\spectrum{\Ldagger_s}=\{-1.1164, 0, 2.0926 ,8.6904\}$, i.e., $\Ldagger_s$ is not psd.
\end{example}

\begin{example}\label{example:normal_L}
The signed Laplacian 
\begin{equation*}
L=\begin{bmatrix} 
0.282&-0.072&0.191&-0.401\\
-0.072&0.252&0.008&-0.189\\
-0.401&-0.189&0.297&0.293\\
0.191&0.008&-0.496&0.297
\end{bmatrix}
\end{equation*}
is normal, and it is $\spectrum{L}=\{0, 0.3983 \pm 0.5920 i , 0.3311\}$, i.e., $-L$ is marginally stable of corank $1$. 
In accordance with Theorem~\ref{thm:directed}, $\spectrum{\Ldagger}=\{0, 0.7823 \pm 1.1628 i , 3.0204\}$, i.e., $-\Ldagger$ is marginally stable of corank $1$, $\spectrum{L_s}=\{0,0.3983,$ $0.3983, 0.3311\}$, i.e., $L_s$ is psd of corank $1$, and $\spectrum{\Ldagger_s}=\{0, 0.7823, 0.7823, 3.0204\}$, i.e., $\Ldagger_s$ is psd of corank $1$.
\end{example}

When the weight balanced digraph is nonnegative, the normality assumption in Theorem~\ref{thm:directed}(iv) can be dropped.
\begin{theorem}\label{thm:directed_nonneg}
	Let $\G(A)$ be a strongly connected nonnegative ($A\ge 0$) digraph such that the corresponding Laplacian $L$ is weight balanced. Let $\Ldagger$ be the irreducible weight balanced pseudoinverse of $L$.
	Then $\Ldagger_s \!=\! \frac{\Ldagger+(\Ldagger)^T}{2}$ is psd of corank $1$.
\end{theorem}
\begin{proof}
	In \cite[Corollary~1]{Altafini2019Investigating} it is shown that when $A\ge 0$, $L_s= \frac{L+L^T}{2}$ is psd of corank $1$ if and only if $\G(A)$ is weight balanced. Using \eqref{eqn:pseudoinv_wb_4} (see Lemma~\ref{lemma:properties_Ldagger}) we can write 
	\begin{align*}
		\Ldagger_s 
		&= \frac{(L+\gamma J)^{-1}+(L^T+\gamma J)^{-1}}{2}-\frac{1}{\gamma} J
		\\
		&= (L+\gamma J)^{-1}\frac{L^T+\gamma J+L+\gamma J}{2}(L^T+\gamma J)^{-1}-\frac{1}{\gamma} J
		\\
		&= (L+\gamma J)^{-1}(L_s+\gamma J)(L^T+\gamma J)^{-1}
		\\&\;\, -\frac{1}{\gamma}(L+\gamma J)^{-1} \bigl( (L+\gamma J) J(L^T+\gamma J)\bigr)(L^T+\gamma J)^{-1}
		\\
		&\overset{\ast}{=} (L+\gamma J)^{-1} (L_s+\gamma J - \gamma J ) (L^T+\gamma J)^{-1}
		\\
		&= (L+\gamma J)^{-1} L_s (L+\gamma J)^{-T},
	\end{align*}
	where in the step marked $\ast$ we have used the properties of $J$ listed in Lemma~\ref{lemma:properties_J}.
	Hence, if $L_s$ is psd (of corank $1$) so is $\Ldagger_s $. 
\end{proof}

\subsection{Undirected signed network case}
Assume that the graph $\G(A) = (\nodeSet,\edgeSet,A) $ is undirected, connected and without self-loops, which means that the adjacency matrix $A$ is irreducible, symmetric and with null-diagonal.
Consequently, the Laplacian $ L= \Sigma -A$ is a symmetric matrix, meaning that $\ker{L^T}= \ker{L}=\vspan{\1}$ and that $\range{L} = \vspan{\1}^\perp$.
In this case, Theorem~\ref{thm:directed} proves the equivalence between eventual exponential positivity of $-L$ and positive semidefiniteness of $L$ (of corank $1$).

For undirected networks the conditions discussed in Theorem~\ref{thm:directed} can be applied to the Kron reduction of $\G(A)$, as described in Section~\ref{sec:Kron_reduction}, showing that if $-L$ is eventually exponentially positive then $-L_r$ (where $L_r$ is the Kron reduced matrix) is also eventually exponentially positive.
\begin{theorem}\label{thm:undirected_kron}
	Let $\G(A)$ be an undirected, connected, signed network with Laplacian $L$. 
	Let $\alpha$ (with $\card{\alpha}\in [2,n-1]$) and $\beta=\{1,\dots,n\} \setminus \alpha$ be a partition of the node set $\nodeSet$.
	Let $\G_r$ be the signed undirected graph obtained by applying the Kron reduction on $\G$, and let $L_r=L/L[\beta]$ be its (symmetric) Laplacian.
	Consider the following conditions: 
	\begin{enumerate}[label=(\roman*)]
	    \item $-L$ is eventually exponentially positive;
		\item $-L_r$ is psd of corank $1$;
		\item $-L_r$ is eventually exponentially positive.
	\end{enumerate}
	If $L$ satisfies (i) then $L_r$ satisfies (ii) and (iii). 
	
    Furthermore, if $\alpha$ is the set of nodes incident to negatively weighted edges and $\beta = \{1,\dots,n\}\setminus \alpha$, then the conditions (i), (ii), (iii) are equivalent.
\end{theorem}
\begin{proof}
    Let $\alpha$ (with $\card{\alpha}\in [2,n-1]$) and $\beta=\{1,\dots,n\} \setminus \alpha$ be a partition of the node set $\nodeSet$ meaning that, after an adequate permutation, $L$ can be rewritten as $L = \begin{bmatrix} L[\alpha] & L[\alpha,\beta] \\L[\beta,\alpha] & L[\beta]\end{bmatrix}$. 
    Let $L_r = L/L[\beta] = L[\alpha]- L[\alpha,\beta] L[\beta]^{-1} L[\beta,\alpha]\in \R^{\card{\alpha}\times \card{\alpha}}$ be the Kron reduced matrix.
    Observe that $L_r$ is symmetric and that $\1_{\card{\alpha}} \in \ker{L_r}$ (see also \cite[Lemma II.1]{DorflerBullo2013}), meaning that $L_r$ is itself a signed Laplacian.
    
    (i)$\Longrightarrow$(ii)$\Longleftrightarrow$(iii).
    Assume that $-L$ is eventually exponentially positive or, equivalently, that $L$ is psd of corank $1$ (see Theorem~\ref{thm:Altafini2019}).
    Then $L[\beta]$ is also psd as it is a principal submatrix of $L$. 
    In what follows we prove (by contradiction) that $L[\beta]$ is actually pd, since $L$ is irreducible and has the row and column inclusion property. 
    
    Let $\card{\beta}=1$ and assume, by contradiction, that $L[\beta]=\Sigma[\beta]=0$. However, $L$ psd means that $L$ has the row and column inclusion property, i.e., if the diagonal element $\Sigma[\beta]$ is zero then $A[\alpha,\beta]=0$ and $A[\beta,\alpha]=0$, which contradicts the hypothesis that $L$ (and $A$) is irreducible. Hence, $L[\beta]>0$ (pd).
    Now we repeat the same argument for $1<\card{\beta}\le n-2$: suppose by contradiction that $L[\beta]$ is psd, i.e., there exists a vector $v\in \R^{\card{\beta}}$ s.t. $L[\beta]v=0$. Then $\bar v= \begin{bmatrix} 0\\v\end{bmatrix}$ is s.t. $L \bar v=0$ (since $\bar v^T L \bar v=0$), which contradicts the hypothesis that $L$ has corank $1$ since $\1 \in \ker{L}$ and $\bar v\notin \vspan{\1}$ (notice that if $v=\1_{\card{\beta}}$, then either $A[\beta,\alpha]$ is the zero matrix - in contradiction with the hypothesis that $L$ is irreducible -, or $\begin{bmatrix} \1_{\card{\alpha}}\\ 0\end{bmatrix},\begin{bmatrix} 0\\ \1_{\card{\beta}}\end{bmatrix}\in \ker{L}$ - in contradiction with the hypothesis that $L$ has corank $1$).
    Therefore, $L[\beta]$ is pd. 
    
    Rewrite $L$ as
    \begin{equation*}
    L = \begin{bmatrix} I & L[\alpha,\beta] L[\beta]^{-1} \\ 0 & I\end{bmatrix} 
    \begin{bmatrix} L_r & 0 \\ 0 & L[\beta]\end{bmatrix}
    \begin{bmatrix} I & 0 \\ L[\beta]^{-1} L[\beta,\alpha] & I\end{bmatrix} 
    \end{equation*}
    where $L[\alpha,\beta] L[\beta]^{-1}=(L[\beta]^{-1} L[\beta,\alpha])^T$.
    Then, applying Sylverster's law of inertia, $L$ psd of corank $1$ and $L[\beta]$ pd imply $L_r$ psd of corank $1$ or, equivalently (from Theorem~\ref{thm:Altafini2019}), $-L_r$ eventually exponentially positive.
    
    (i)$\Longleftrightarrow$(ii)$\Longleftrightarrow$(iii). Let $\alpha$ be the set of nodes incident to negatively weighted edges.
	In what follows, the steps marked by the symbol $\star$ follow from Theorem~\ref{thm:Altafini2019} while the step marked by the symbol $\triangle$ from \cite[Theorem~1]{Chen2020}. 
	\begin{align*}
		&-L \text{ eventually exponentially positive}
		\\ &\overset{\star}{\Leftrightarrow}
		L \text{ psd of corank $1$}
		\\ &
		\overset{\triangle}{\Leftrightarrow}
		L_r \text{ psd of corank $1$}
		\\ &\overset{\star}{\Leftrightarrow}
		-L_r \text{ eventually exponentially positive}.
	\end{align*}
\end{proof}

Similarly to Corollary~\ref{cor:class_exppos_directed}, from Theorems~\ref{thm:directed} and \ref{thm:undirected_kron} we obtain the following characterization of the class of eventually exponentially positive Laplacian matrices of undirected graphs.
\begin{corollary}
The class of eventually exponentially positive, irreducible, symmetric Laplacian matrices is closed under the pseudoinverse operation and the operation of Kron reduction.
\end{corollary}

\section{ELECTRICAL NETWORKS AND EFFECTIVE RESISTANCE}
A resistive electrical network can be represented as a graph $\G(A)=(\nodeSet,\edgeSet,A)$ where each weight $a_{ij}$ represents the inverse of the resistance between the two nodes (i.e., the conductance of the transmission): $a_{ij}=\frac{1}{r_{ij}}$, see \cite{KleinRandic1993,GhoshBoydSaberi2008} and \cite{DorflerPorcoBullo2018} for an overview. 
The notion of effective resistance between a pair of nodes (see e.g. \cite{DorflerPorcoBullo2018}) is related to the pseudoinverse of the Laplacian associated to the electrical network. When the network is connected, undirected and nonnegative, its Laplacian (and its pseudoinverse) is known to be psd of corank $1$, which means that the effective resistance between two nodes is well-defined (see e.g. \cite{GhoshBoydSaberi2008} for its properties). 
Extensions to signed graphs and negative resistances have been investigated in \cite{ZelazoBurger2014,Chen2016a,ZelazoBurger2017,ChenAl2016,Chen2020}, where positive semidefiniteness of the Laplacian is expressed in terms of effective resistance.

In what follows we make use of $\Ldagger_s$ to extend the notion of effective resistance to directed (strongly connected) signed networks whose Laplacian $L$ is a normal matrix and $-L$ is eventually exponentially positive.
\begin{definition}
	The effective resistance between two nodes $i,j\in \{1,\dots,n\}$ of a signed digraph whose corresponding Laplacian $L$ is normal and $-L$ is eventually exponentially positive, is given by
	\begin{equation}
		R_{ij} 	= (e_i-e_j)^T \Ldagger_s (e_i-e_j),
		\label{eqn:Rij}
	\end{equation}
	where $\Ldagger_s = \frac{\Ldagger+(\Ldagger)^T}{2}$ and $\Ldagger$ is the pseudoinverse of $L$.
	The effective resistance matrix $R=[R_{ij}]$ is defined as
	\begin{equation}
		R= D_{\Ldagger_s} \1 \1^T + \1 \1^T D_{\Ldagger_s} - 2 \Ldagger_s
		\label{eqn:R}
	\end{equation}
	where $D_{\Ldagger_s} = \diag{[\Ldagger_s]_{11} , \dots, [\Ldagger_s]_{nn}}$ is a diagonal matrix whose elements are the diagonal elements of $\Ldagger_s$.
	The total effective resistance is defined as
	\begin{equation}
		\Rtot = \frac{1}{2}\1^T R \1.
		\label{eqn:Rtot}
	\end{equation}
\end{definition}
In the literature on undirected networks, the total effective resistance \eqref{eqn:Rtot} is also called "weighted effective graph resistance" \cite{Ellens2011} or "Kirchhoff index" \cite{XiaoGutman2003}, and represents the overall transport capability of the graph \cite{VanMieghem2017}.
\begin{remark}
	If the graph is undirected, eq.~\eqref{eqn:Rij} reduces to the standard notion of effective resistance since $\Ldagger_s = \Ldagger$.	
\end{remark}
The effective resistance \eqref{eqn:Rij}, as its counterpart for undirected graphs (see \cite{KleinRandic1993,GhoshBoydSaberi2008,DorflerPorcoBullo2018}), is still nonnegative and symmetric, its square root is a metric, and the effective resistance matrix \eqref{eqn:R} is a Euclidean distance matrix, i.e., it has nonnegative elements, zero diagonal elements and it is negative semidefinite on $\1^\perp$ \cite{GhoshBoydSaberi2008}.
The last part of the proof of the next lemma follows \cite[Section 2.8]{GhoshBoydSaberi2008} and is here reported for completeness.
\begin{lemma}\label{lemma:Rmetric}
    The square root of the effective resistance \eqref{eqn:Rij} between two nodes $i,j\in \{1,\dots,n\}$ of a signed digraph with normal Laplacian $L$ is a metric: it is nonnegative, symmetric and it satisfies the triangle inequality.
    The effective resistance matrix \eqref{eqn:R} is a Euclidean distance matrix.
\end{lemma}
\begin{proof}
    Theorem~\ref{thm:directed} shows that for a signed digraph with normal Laplacian $L$ s.t. $-L$ is eventually exponentially positive, the matrix $\Ldagger_s$ is itself a signed Laplacian and it is psd of corank $1$ with $\ker{\Ldagger_s}=\vspan{\1}$.
    Since $R_{ij}$ is a quadratic form generated by $\Ldagger_s$, then 
    \begin{align*}
        R_{ij}&=(e_i-e_j)^T \Ldagger_s (e_i-e_j)
        = \norm{ (\Ldagger_s)^\half (e_i-e_j)}_2^2
        \\&= \norm{ (\Ldagger_s)^\half (e_j-e_i)}_2^2
        =(e_j-e_i)^T \Ldagger_s (e_j-e_i)
        = R_{ji}
    \end{align*}
    for all $i,j=1,\dots,n$, and
    \begin{align*}
        R_{ij}&=(e_i-e_j)^T \Ldagger_s (e_i-e_j)
        = \norm{ (\Ldagger_s)^\half (e_i-e_j)}_2^2\ge 0
    \end{align*}
    for all $i,j=1,\dots,n$, with $R_{ij}=0$ if and only if $i=j$ (since $e_i-e_j \in \vspan{\1^\perp}$ when $i\ne j$).
    Moreover, 
    \begin{align*}
        \sqrt{R_{ik}} + \sqrt{R_{kj}}
        &= \norm{ (\Ldagger_s)^\half (e_i-e_k)}_2 +\norm{ (\Ldagger_s)^\half(e_k-e_j)}_2
        \\ 
        &\ge \norm{ (\Ldagger_s)^\half (e_i-e_k)+(\Ldagger_s)^\half(e_k-e_j)}_2
        \\ 
        &= \norm{ (\Ldagger_s)^\half (e_i-e_j)}_2 = \sqrt{R_{ij}}
    \end{align*} 
    for all $i,j,k=1,\dots,n$, i.e., the triangle inequality holds.
    
    Finally, to prove that $R$ is an Euclidean distance matrix we need to show that $x^T R x \le 0$ for all $x\perp \1$:
    \begin{equation*}
        x^T R x = x^T (D_{\Ldagger_s} \1 \1^T  + \1 \1^T D_{\Ldagger_s} - 2 \Ldagger_s) x
         = - 2 x^T \Ldagger_s x \le 0,
    \end{equation*}
    since $\Ldagger_s$ is psd with $\ker{\Ldagger_s}=\vspan{\1}$.
\end{proof}

Notice that if we consider only nonnegative digraphs then the normality assumption of the Laplacian can be replaced by the less restrictive weight balanced assumption when defining the effective resistance in \eqref{eqn:Rij}.
Indeed, Theorem~\ref{thm:directed_nonneg} shows that if the digraph is nonnegative and strongly connected then $\Ldagger_s$ is psd of corank $1$.

\begin{proposition}
	Consider a nonnegative strongly connected weight balanced digraph $\G(A)$ (with $A\ge0$). Then $R_{ij} \ge 0$ for all $i,j=1,\dots,n$, and $\Rtot\ge 0$.
\end{proposition}

Another generalization of the notion of effective resistance for directed, strongly connected, nonnegative networks is introduced in \cite{YoungScardoviLeonard2010,YoungScardoviLeonard2016a}.
The authors use the fact that the Laplacian $L$ is marginally stable and its projection on $\1^\perp$, denoted $\bar L= QLQ^T$ (where the rows of $Q\in \R^{n-1\times n}$ form an orthonormal basis for $\1^\perp$), is Hurwitz stable, to define the effective resistance between nodes $i$ and $j$ as $\tilde{R}_{ij}=(e_i-e_j)^T X (e_i-e_j)$, where $X= 2 Q^T S Q$ and $S$ is the pd solution of the Lyapunov equation $\bar L S + S \bar L^T=I_{n-1}$. 
The Kirchhoff index is then defined as $K_f=\sum_{i<j} \tilde{R}_{ij}$.

If we consider digraphs $\G(A)$ whose Laplacian is a normal matrix, $K_f$ reduces to $K_f = n \sum_{i=2}^n \frac{1}{\real{\lambda_i(L)}}$ and we can show that it provides an upper bound for $\Rtot$ defined in \eqref{eqn:Rtot}.
\begin{align*}
&\Rtot 
= n \cdot \text{trace}(\Ldagger_s) 
= n \cdot \sum_{i=2}^n \lambda_i(\Ldagger_s)    
\\&\;= n \cdot \sum_{i=2}^n \real{\lambda_i(\Ldagger)}
= n \cdot \sum_{i=2}^n \real{\frac{1}{\lambda_i(L)}}
\\&\;= n \cdot \sum_{i=2}^n \frac{\real{\lambda_i(L)}}{\abs{\lambda_i(L)}^2}
= n \cdot \sum_{i=2}^n \frac{1}{\real{\lambda_i(L)}\bigl(1+\frac{\imag{\lambda_i(L)}^2}{\real{\lambda_i(L)}^2}\bigr)}
\\&\;\le n \cdot \sum_{i=2}^n \frac{1}{\real{\lambda_i(L)}}
= K_f,
\end{align*}
with equality only if $\G(A)$ is undirected (notice that $L$ normal and non-symmetric means $n\ge 3$).

\begin{example}
Let $\G(A)$ be a nonnegative, unweighted, directed, cycle graph, whose Laplacian $L$ is a normal matrix with eigenvalues $1\!+\!e^{i \theta_k}$, with $\theta_k=\pi \bigl(1-\frac{2k}{n}\bigr)$, for all $k=0,\ldots,n-1$.
Then, $K_f = \frac{n(n^2-1)}{6}$ (see e.g. \cite{YoungScardoviLeonard2010}), $\Rtot 
= n \cdotp \sum_{k=2}^n \real{\frac{1}{\lambda_k(L)}}
= n \cdotp \sum_{k=2}^n \frac{1+\cos \theta_k}{(1+\cos \theta_k)^2+\sin^2\theta_k}
= n \cdotp \sum_{k=2}^n \frac{1}{2}= \frac{n (n-1)}{2}$, and we obtain $\Rtot \le \! K_f$ for all $n\ge 2$.
\end{example}

\section{CONCLUSIONS AND FUTURE WORK}
For signed Laplacians which are weight balanced, marginal stability (of corank $1$) is equivalent to eventual exponential stability. 
This work shows that the class of eventually exponentially positive, weight balanced Laplacians is closed under the pseudoinverse operation and, therefore, it provides a natural embedding for the usual nonnegative Laplacian. As a byproduct we get conditions for checking the marginal stability of the pseudoinverse of signed Laplacians.
Moreover, closure under the symmetrization operation can be proven when this class is restricted to Laplacians that are also normal matrices. The normality assumption is a sufficient condition and it remains to be investigated if it can be relaxed.

In addition, we would like to gain a better understanding of the set of eventually exponentially positive, weight balanced Laplacians and its properties. For instance, it is easy to observe that it is not a convex cone, not even if we consider normal matrices (but the intuition is that this set is actually a convex cone, without the origin, if we restrict to undirected graphs). However, similarly to \cite{JohnsonTarazaga2004}, it is possible to show that it is path-wise connected. These considerations, among other directions, will be investigated in a future paper.


\begin{thebibliography}{10}
	\providecommand{\url}[1]{#1}
	\csname url@samestyle\endcsname
	\providecommand{\newblock}{\relax}
	\providecommand{\bibinfo}[2]{#2}
	\providecommand{\BIBentrySTDinterwordspacing}{\spaceskip=0pt\relax}
	\providecommand{\BIBentryALTinterwordstretchfactor}{4}
	\providecommand{\BIBentryALTinterwordspacing}{\spaceskip=\fontdimen2\font plus
		\BIBentryALTinterwordstretchfactor\fontdimen3\font minus
		\fontdimen4\font\relax}
	\providecommand{\BIBforeignlanguage}[2]{{%
			\expandafter\ifx\csname l@#1\endcsname\relax
			\typeout{** WARNING: IEEEtran.bst: No hyphenation pattern has been}%
			\typeout{** loaded for the language `#1'. Using the pattern for}%
			\typeout{** the default language instead.}%
			\else
			\language=\csname l@#1\endcsname
			\fi
			#2}}
	\providecommand{\BIBdecl}{\relax}
	\BIBdecl
	
	\bibitem{Chung1997}
	F.~R.~K. Chung, \emph{{Spectral Graph Theory}}, ser. CBMS Number 92.\hskip 1em
	plus 0.5em minus 0.4em\relax American Mathematical Society, 1997.
	
	\bibitem{Agaev2005Spectra}
	R.~Agaev and P.~Chebotarev, ``{On the spectra of nonsymmetric Laplacian
		matrices},'' \emph{Lin. Algebra Appl.}, vol. 399, pp. 157--168, 2005.
	
	\bibitem{OlfatiSaberMurray2004}
	R.~Olfati-Saber and R.~Murray, ``{Consensus Problems in Networks of Agents With
		Switching Topology and Time-Delays},'' \emph{IEEE TAC}, vol.~49, no.~9, pp.
	1520--1533, sep 2004.
	
	\bibitem{AltafiniLini2015}
	C.~Altafini and G.~Lini, ``{Predictable dynamics of opinion forming for
		networks with antagonistic interactions},'' \emph{IEEE TAC}, vol.~60, no.~2,
	pp. 342--357, feb 2015.
	
	\bibitem{BronskiDeville2014}
	J.~C. Bronski and L.~DeVille, ``{Spectral Theory for Dynamics on Graphs
		Containing Attractive and Repulsive Interactions},'' \emph{SIAM J. Appl.
		Math.}, vol.~74, no.~1, pp. 83--105, jan 2014.
	
	\bibitem{PanShaoMesbahi2016}
	L.~Pan, H.~Shao, and M.~Mesbahi, ``{Laplacian dynamics on signed networks},''
	in \emph{55th IEEE CDC}, Las Vegas, USA, dec 2016, pp. 891--896.
	
	\bibitem{KleinRandic1993}
	D.~J. Klein and M.~Randi{\'{c}}, ``{Resistance distance},'' \emph{Journal of
		Mathematical Chemistry}, vol.~12, no.~1, pp. 81--95, 1993.
	
	\bibitem{XiaoGutman2003}
	W.~Xiao and I.~Gutman, ``{Resistance distance and Laplacian spectrum},''
	\emph{Theoretical Chemistry Accounts}, vol. 110, no.~4, pp. 284--289, 2003.
	
	\bibitem{GhoshBoydSaberi2008}
	A.~Ghosh, S.~Boyd, and A.~Saberi, ``{Minimizing Effective Resistance of a
		Graph},'' \emph{SIAM Review}, vol.~50, no.~1, pp. 37--66, jan 2008.
	
	\bibitem{Chandra1996}
	A.~K. Chandra \emph{et~al.}, ``{The electrical resistance of a graph captures
		its commute and cover times},'' \emph{Computational Complexity}, vol.~6,
	no.~4, pp. 312--340, 1996.
	
	\bibitem{Palacios2001}
	J.~L. Palacios, ``{Resistance distance in graphs and random walks},''
	\emph{Int. J. Quantum Chem.}, vol.~81, no.~1, pp. 29--33, 2001.
	
	\bibitem{Boley2011}
	D.~Boley, G.~Ranjan, and Z.~L. Zhang, ``{Commute times for a directed graph
		using an asymmetric Laplacian},'' \emph{Lin. Algebra Appl.}, vol. 435, no.~2,
	pp. 224--242, 2011.
	
	\bibitem{VanMieghem2017}
	P.~{Van Mieghem}, K.~Devriendt, and H.~Cetinay, ``{Pseudoinverse of the
		Laplacian and best spreader node in a network},'' \emph{Physical Review E},
	vol.~96, no.~3, pp. 1--22, 2017.
	
	\bibitem{YoungScardoviLeonard2010}
	G.~F. Young, L.~Scardovi, and N.~E. Leonard, ``{Robustness of noisy consensus
		dynamics with directed communication},'' in \emph{2010 ACC}, Baltimore, MD,
	USA, jun 2010, pp. 6312--6317.
	
	\bibitem{YoungScardoviLeonard2011}
	------, ``{Rearranging trees for robust consensus},'' in \emph{50th IEEE CDC
		and ECC}, Orlando, FL, USA, 2011, pp. 1000--1005.
	
	\bibitem{lindmark2020investigating}
	G.~Lindmark and C.~Altafini, ``{Investigating the effect of edge modifications
		on networked control systems},'' \emph{arXiv:2007.13713}, 2020.
	
	\bibitem{Noutsos2006}
	D.~Noutsos, ``{On Perron–Frobenius property of matrices having some negative
		entries},'' \emph{Lin. Algebra Appl.}, vol. 412, no. 2-3, pp. 132--153, 2006.
	
	\bibitem{NoutsosTsatsomeros2008}
	D.~Noutsos and M.~J. Tsatsomeros, ``{Reachability and Holdability of
		Nonnegative States},'' \emph{SIAM J. Matrix Analysis Appl.}, vol.~30, no.~2,
	pp. 700--712, jan 2008.
	
	\bibitem{JohnsonTarazaga2004}
	C.~R. Johnson and P.~Tarazaga, ``{On matrices with Perron-Frobenius properties
		and some negative entries},'' \emph{Positivity}, vol.~8, no.~4, pp. 327--338,
	2004.
	
	\bibitem{Shi2019Dynamics}
	G.~Shi, C.~Altafini, and J.~S. Baras, ``{Dynamics over Signed Networks},''
	\emph{SIAM Review}, vol.~61, no.~2, pp. 229--257, jan 2019.
	
	\bibitem{YoungScardoviLeonard2016a}
	G.~F. Young, L.~Scardovi, and N.~E. Leonard, ``{A New Notion of Effective
		Resistance for Directed Graphs—Part I: Definition and Properties},''
	\emph{IEEE TAC}, vol.~61, no.~7, pp. 1727--1736, jul 2016.
	
	\bibitem{Meyer2000}
	C.~D. Meyer, \emph{{Matrix Analysis and Applied Linear Algebra}}.\hskip 1em
	plus 0.5em minus 0.4em\relax Society for Industrial {\&} Applied Mathematics,
	2000.
	
	\bibitem{HornJohnson2013}
	R.~A. Horn and C.~R. Johnson, \emph{{Matrix analysis}}, 2nd~ed.\hskip 1em plus
	0.5em minus 0.4em\relax Cambridge University Press, 2013.
	
	\bibitem{DorflerBullo2013}
	F.~D{\"{o}}rfler and F.~Bullo, ``{Kron reduction of graphs with applications to
		electrical networks},'' \emph{IEEE Trans. Circuits Syst. I: Regular Papers},
	vol.~60, no.~1, pp. 150--163, 2013.
	
	\bibitem{DorflerPorcoBullo2018}
	F.~Dorfler, J.~W. Simpson-Porco, and F.~Bullo, ``{Electrical Networks and
		Algebraic Graph Theory: Models, Properties, and Applications},''
	\emph{Proceedings of the IEEE}, vol. 106, no.~5, pp. 977--1005, 2018.
	
	\bibitem{Chen2016a}
	Y.~Chen, S.~Z. Khong, and T.~T. Georgiou, ``{On the definiteness of graph
		Laplacians with negative weights: Geometrical and passivity-based
		approaches},'' in \emph{2016 ACC}, Boston, MA, USA, jul 2016, pp. 2488--2493.
	
	\bibitem{Altafini2019Investigating}
	C.~Altafini, ``{Investigating stability of Laplacians on signed digraphs via
		eventual positivity},'' in \emph{58th IEEE CDC}, Nice, France, dec 2019, pp.
	5044--5049.
	
	\bibitem{LewisNewman1968}
	T.~O. Lewis and T.~G. Newman, ``{Pseudoinverses of Positive Semidefinite
		Matrices},'' \emph{SIAM J. Appl. Math.}, vol.~16, no.~4, pp. 701--703, 1968.
	
	\bibitem{Bullo2020Lectures}
	F.~Bullo, \emph{{Lectures on Nonlinear Network Systems (ed. 1.4)}}.\hskip 1em
	plus 0.5em minus 0.4em\relax Kindle Direct Publishing, 2020.
	\url{http://motion.me.ucsb.edu/book-lns}
	
	\bibitem{BenIsrealGreville2003}
	A.~Ben-Israel and T.~N.~E. Greville, \emph{{Generalized Inverses}}, 2nd~ed.,
	ser. CMS Books in Mathematics.\hskip 1em plus 0.5em minus 0.4em\relax New
	York: Springer-Verlag, 2003.
	
	\bibitem{Chen2020}
	W.~Chen \emph{et~al.}, ``{On Spectral Properties of Signed Laplacians with
		Connections to Eventual Positivity},'' \emph{IEEE TAC}, vol.~66, no.~5, pp.
	2177--2190, 2020.
	
	\bibitem{ZelazoBurger2014}
	D.~Zelazo and M.~Burger, ``{On the definiteness of the weighted Laplacian and
		its connection to effective resistance},'' in \emph{53rd IEEE CDC}, dec 2014,
	pp. 2895--2900.
	
	\bibitem{ZelazoBurger2017}
	------, ``{On the robustness of uncertain consensus networks},'' \emph{IEEE
		Trans. Control Netw. Syst.}, vol.~4, no.~2, pp. 170--178, 2017.
	
	\bibitem{ChenAl2016}
	W.~Chen \emph{et~al.}, ``{Characterizing the positive semidefiniteness of
		signed Laplacians via Effective Resistances},'' in \emph{55th IEEE CDC}, Las
	Vegas, USA, dec 2016, pp. 985--990.
	
	\bibitem{Ellens2011}
	W.~Ellens \emph{et~al.}, ``{Effective graph resistance},'' \emph{Lin. Algebra
		Appl.}, vol. 435, no.~10, pp. 2491--2506, 2011.
	
\end{thebibliography}

\end{document}